\documentclass[journal]{IEEEtran}
\usepackage{graphicx} 
\usepackage{amsmath}
\usepackage{amssymb}
\usepackage{amsthm}

\newtheorem{theorem}{Theorem}
\newtheorem{remark}{Remark}

\usepackage{multicol}
\usepackage{multirow}
\usepackage[english]{babel}   
\usepackage{color}
\usepackage{multicol}
\usepackage[T1]{fontenc}
\usepackage{epstopdf}
\ifCLASSINFOpdf
\else
\fi
\begin{document}
%
\title{Theory and Application on Adaptive-Robust Control of Euler-Lagrange Systems with Linearly Parametrizable Uncertainty Bound}
%
%
%

\author{Spandan~Roy,~Sayan Basu Roy
        and~Indra~Narayan~Kar,~\IEEEmembership{Senior~Member,~IEEE}
\thanks{S. Roy, S. Basu Roy and I. N. Kar are with the Department
of Electrical Engineering, Indian Institute of Technology-Delhi, New Delhi,
India e-mail: (sroy002@gmail.com, sayanetce@gmail.com, ink@ee.iitd.ac.in).}}
\maketitle

\begin{abstract}

This work proposes a new adaptive-robust control (ARC) architecture for a class of uncertain Euler-Lagrange (EL) systems where the upper bound of the uncertainty satisfies linear in parameters (LIP) structure. Conventional ARC strategies either require structural knowledge of the system or presume that the overall uncertainties or its time derivative is norm bounded by a constant. Due to unmodelled dynamics and modelling imperfection, true structural knowledge of the system is not always available. Further, for the class of systems under consideration, prior assumption regarding the uncertainties (or its time derivative) being upper bounded by a constant, puts a restriction on states beforehand. Conventional ARC laws invite overestimation-underestimation problem of switching gain. Towards this front, Adaptive Switching-gain based Robust Control (ASRC) is proposed which alleviates the overestimation-underestimation problem of switching gain. Moreover, ASRC avoids any presumption of constant upper bound on the overall uncertainties and can negotiate uncertainties regardless of being linear or nonlinear in parameters. Experimental results of ASRC using a wheeled mobile robot notes improved control performance in comparison to adaptive sliding mode control.
\end{abstract}

\begin{IEEEkeywords}
Adaptive-robust control, Euler-Lagrange systems, Wheeled mobile robot, Uncertainty.
\end{IEEEkeywords}

\section{Introduction}
%
%
%
%
\subsection{Background}
\IEEEPARstart{T}{he} controller design aspect for nonlinear systems subjected to parametric and nonparametric uncertainties has always been a challenging task. Adaptive control and Robust control are the two popular control strategies to deal with uncertain nonlinear systems. In case of adaptive control, online computation of the unknown system parameters and controller gains for complex systems is significantly intensive \cite{Ref:2}. On the other front, robust control reduces computation burden for complex systems compared to adaptive control, while requiring a predefined upper bound on the uncertainties. 
However, in practice it is not always possible to estimate a prior uncertainty bound due to the effect of unmodelled dynamics. Again, to increase the operating region of the controller, often higher uncertainty bounds are assumed. This in turn leads to overestimation of switching gain and high control effort \cite{Ref:23}.
\par Considering the individual constraints of adaptive and robust control, recently global research is reoriented towards adaptive-robust control (ARC). The series of publications \cite{Ref:2}, \cite{Ref:13}-\cite{Ref:19} regarding ARC, estimate the individual uncertain system parameters through adaptive law and robust control is utilized to negate the effect of external disturbances. These works utilize the projection operator in their respective adaptive laws which necessitate the knowledge of lower and upper bound of individual uncertain system parameters. Adaptive sliding mode control (ASMC) is designed in \cite{Ref:32-2} for parameter identification of mechanical servo systems with LuGre friction considering the uncertainties to be linear in parameters (LIP). In contrast, the controllers \cite{Ref:20}-\cite{Ref:frid} assume that the overall uncertainty (or its time derivative) is bounded by some constant. Thereafter, that constant term is estimated by adaptive law, rather estimating individual uncertain system parameters. The adaptive laws in \cite{Ref:21}-\cite{Ref:22} involve a predefined threshold value; as a matter of fact, until the threshold value is achieved, the switching gain may still be increasing (resp. decreasing) even if the tracking error decreases (resp. increases) and thus creates overestimation (resp. underestimation) problem of switching gain.{While the underestimation problem compromises the controller accuracy by applying lower switching gain than the required amount, the overestimation problem causes larger gain and high control input requirement. 
The adaptive law reported in \cite{Ref:ut} requires predefined bound on the time derivative of the uncertainties. As observed in \cite{Ref:frid}, the method in \cite{Ref:ut} also requires frequency characteristics of the perturbation to design the filter for equivalent control. However, the work in \cite{Ref:frid} assumed that the time derivative of the uncertainties are bounded by an unknown constant. 
\subsection{Motivation}
Let us consider the following system representing a chemostat operating under Monod kinetic \cite{Ref:bio}:
\begin{align}
&\dot{z}_1=f_1(z_1,z_2)-Dz_1,\dot{z}_2=f_2(z_1,z_2)+S_0-Dz_2, \label{b}\\
&\text{where}~f_1(z_1,z_2)=\frac{\delta_1 z_1 z_2}{\delta_2+z_2}, f_2(z_1,z_2)=\frac{-\delta_3 z_1 z_2}{\delta_2+z_2}. \nonumber
\end{align}
Here $z_1 \geq 0, z_2 \geq 0$ $\forall t \geq 0$ are states; $\delta_1, \delta_2, \delta_3$ are uncertain positive parameters; $S_0$ is a known constant and $D$ is the control input. For the system (\ref{b}), the following relations hold:
\begin{align}
&|f_1| \leq y_1 f(z), |f_2| \leq y_2f(z), \label{b1} \\
&\text{where}~y_1=|\delta_1|/|\delta_2|, y_2=|\delta_3|/|\delta_2|, f(z)=|z_1||z_2|. \nonumber
\end{align}
Inspection of the uncertainties $f_1$ and $f_2$ from (\ref{b})-(\ref{b1}) reveal that though $f_1,f_2$ are nonlinear in parameters (NLIP) but their upper bounds are LIP.  

Similarly Euler-Lagrange (EL) systems can have uncertainties with LIP or NLIP (e.g. system with nonlinear friction (\cite{Ref:nlip})) structure. However, the upper bound of the overall (or lumped) uncertainty for such systems has LIP property \cite{Ref:spong}. EL systems, in general, represent a large class of real world systems like robotic manipulators \cite{Ref:n1}-\cite{Ref:n2}, mobile robots \cite{Ref:self}, ship dynamics, aircraft, pneumatic muscles \cite{Ref:n3} etc. These systems have immense applications in various domains such as defence, automation industry, surveillance, space missions etc.  
The controllers \cite{Ref:20}-\cite{Ref:22} assume that the overall uncertainties are upper bounded by some constant, while \cite{Ref:ut}-\cite{Ref:frid} assume the time derivative of the overall uncertainty to be bounded by some constant. Hence, for the aforementioned class of systems, consideration of such constant bound (known or unknown) restricts the system state a priori. Further, the switching gain in \cite{Ref:21}-\cite{Ref:22} suffers from over- and under-estimation problems. In practice, it is also not always possible to have prior knowledge of bounds for system parameters as required in \cite{Ref:2}, \cite{Ref:13}-\cite{Ref:19} for projection operator. 
\subsection{Contribution}\label{sec 1.2}
In view of the above discussion and the importance of EL systems in real-life scenarios, it is imperative to formulate a dedicated ARC framework for uncertain EL systems. Towards this front, 
Adaptive Switching-gain based Robust Control (ASRC) is presented in this paper for tracking control of uncertain EL systems. The formulation of ASRC is insensitive towards the nature of the uncertainties, i.e., it can negotiate uncertainties that can be either LIP or NLIP. ASRC utilizes \textbf{LIP structure of the upper bound of uncertainty} and does not presume the overall uncertainty (or its time derivative) to be upper bounded by a constant. The adaptive law of ASRC prevents the switching gain from becoming a monotonically increasing function by allowing the switching gain to decrease within a finite time when tracking error decreases. Moreover, ASRC alleviates the overestimation-underestimation problem of switching gain. To realize the effectiveness, the performance of ASRC is compared with ASMC \cite{Ref:21}-\cite{Ref:22} experimentally using PIONEER 3 wheeled mobile robot (WMR).  


\subsection{Organization and Notations}
The remainder of the article is organized as follows: The proposed ASRC framework for second order EL systems is detailed in Section II. This is followed by the experimental results of ASRC and its comparison with ASMC \cite{Ref:21}-\cite{Ref:22} in Section III. Section IV presents concluding remarks.
\par The following notations are used in this paper: $\lambda _{\min}(\bullet)$ and $|| \bullet ||$ represent minimum eigenvalue and Euclidean norm of $(\bullet)$ respectively; $I$ denotes identity matrix with appropriate dimension; $\mathbb{R}^{+}$ denotes the set of positive real numbers.
\section{Controller Design}
\subsection{Problem Formulation}
In general, an EL system with second order dynamics can be written as
\begin{equation}\label{sys}
M(q)\ddot{q}+C(q,\dot{q})\dot{q}+g(q)+f(\dot{q})+d_s=\tau,
\end{equation}
where $q\in\mathbb{R}^{n}$ denotes system state, $\tau\in\mathbb{R}^{n}$ denotes vector of generalized control input forces, $M(q)\in\mathbb{R}^{n\times n}$ represents mass/inertia matrix, $C(q,\dot{q})\in\mathbb{R}^{n\times n}$ denotes Coriolis, centripetal terms, $g(q)\in\mathbb{R}^{n}$ denotes gravity vector, $f(\dot{q})\in\mathbb{R}^{n}$ represents the vector of slip, damping and friction forces and $d_s(t)$ denotes the bounded external disturbances. The system (\ref{sys}) possesses the following properties \cite{Ref:spong}: \\ 
\textbf{Property 1:} 
The matrix $(\dot{M}-2{C})$ is skew symmetric. \\
\textbf{Property 2:} $\exists g_b,f_b,\bar{d} \in \mathbb{R}^{+}$ such that $||g(q)|| \leq g_b$, $||f(\dot{q})|| \leq f_b||\dot{q}||$ and $||d_s(t)|| \leq \bar{d}$.\\
\textbf{Property 3:} The matrix $M(q)$ is uniformly positive definite and there exist two positive constants $\mu_1, \mu_2$ such that
\begin{equation}\label{prop 3}
0 < \mu_1 I \leq M(q) \leq \mu_2 I .
\end{equation}
\textbf{Property 4:} $\exists C_b \in \mathbb{R}^{+}$ such that $||C(q,\dot{q})|| \leq C_b ||\dot{q}||$. 

Let $q^d(t)$ is the desired trajectory to be tracked and it is selected such that $q^d, \dot{q}^d, \ddot{q}^d \in \mathcal{L}_{\infty}$.
Let $e(t) \triangleq q(t)-q^d(t)$ be the tracking error and $e_f$ be the filtered tracking error:
\begin{equation}
e_f \triangleq \dot{e}+ \Omega e \Rightarrow e_f= \Gamma \xi, \label{r}
\end{equation}
where $\Gamma \triangleq [\Omega~ I]$, $\xi \triangleq [e^T ~ \dot{e}^T]^T$ and $\Omega \in \mathbb{R}^{n \times n}$ is a positive definite matrix. Multiplying the time derivative of (\ref{r}) by $M$ and using (\ref{sys}) yields
\begin{align}
M\dot{e}_f&=M(\ddot{q}-\ddot{q}^d+\Omega \dot{e})= \tau-C(q,\dot{q})e_f+\sigma, \label{mr dot}
\end{align}
where $\sigma \triangleq -(C(q,\dot{q})\dot{q}+g(q)+f(\dot{q})+d_s+M\ddot{q}^d-M\Omega \dot{e}-C(q,\dot{q})e_f)$ represents the overall uncertainty. Further $\xi=[e^T ~ \dot{e}^T]^T $ implies $||\xi|| \geq ||e|| , ||\xi|| \geq || \dot{e} || $. 

\textbf{Characterization of the upper bound of $\sigma$:} 
Relation (\ref{r}) and system Property 4 yields
\begin{align}
&||C e_f -  C\dot{q}|| =||C(\dot{e}+ \Omega e)-C\dot{q}||=|| - C\dot{q}^d +C\Omega e  || \nonumber\\
& \leq  C_b ||\dot{q}|| || \dot{q}^d|| + C_b ||\dot{q}||||\Omega|| ||e|| \nonumber \\
& \leq C_b ||\dot{e}+\dot{q}^d || || \dot{q}^d|| + C_b ||\dot{e}+\dot{q}^d || ||\Omega || ||\xi ||  \nonumber \\
& \leq C_b \lbrace ||\xi|||| \dot{q}^d|| + || \dot{q}^d||^2+  ||\xi||^2 ||\Omega ||+ || \dot{q}^d || ||\Omega || ||\xi || \rbrace. \label{i}
\end{align}
Further, system Properties 2 and 3 provide the following:
\begin{align}
&||g(q)+f(\dot{q})+d_s+M\ddot{q}^d-M\Omega \dot{e}|| \nonumber\\
&\leq g_b+f_b||\dot{q}||+\bar{d}+\mu_2||\ddot{q}^d||+\mu_2 ||\Omega|| ||\dot{e}|| \nonumber\\
&\leq g_b+f_b||\dot{e}+\dot{q}^d||+\bar{d}+\mu_2||\ddot{q}^d||+\mu_2 ||\Omega|| ||\xi|| \nonumber\\
& \leq g_b+f_b||\xi||+f_b||\dot{q}^d||+\bar{d}+\mu_2||\ddot{q}^d||+\mu_2 ||\Omega|| ||\xi||. \label{ii}
\end{align}
Since $q^d, \dot{q}^d, \ddot{q}^d \in \mathcal{L}_{\infty}$, it can be verified using (\ref{i})-(\ref{ii}) that $\exists \theta_i^{*} \in \mathbb{R}^{+}$ $i=0,1,2$ such that the upper bound of $\sigma$ holds the following LIP structure \cite{Ref:spong}:
\begin{equation}
||\sigma ||\leq \theta_0^{*}+\theta_1^{*}||\xi||+ \theta_2^{*}||\xi||^2  \triangleq Y(\xi)^T \Theta^{*}, \label{sigma1} 
\end{equation}
where $Y(\xi)=[1 ~ ||\xi|| ~ ||\xi||^2]^T$ and $\Theta^{*}=[\theta_0^{*} ~ \theta_1^{*} ~ \theta_2^{*}]^T$.

Let $\bar{\Theta} \triangleq \lbrace {\Theta} \in \mathbb{R}^{3} : {\theta}_i \geq \theta_i^{*} ~\forall i=0,1,2 \rbrace$ such that the following condition always holds from (\ref{sigma1}):
\begin{equation}\label{bar gamma}
||\sigma|| \leq Y(\xi)^{T}{\Theta}, ~~ \forall \Theta \in \bar{\Theta}.
\end{equation}
A robust controller for the system (\ref{sys}) can be designed as \cite{Ref:spong}
\begin{align} 
&\tau =-e-G e_f-\Delta \tau, ~ \Delta \tau=\begin{cases}
    {\rho}\frac{e_f}{|| e_f||}       & ~ \text{if } || e_f|| \geq \varpi\\
    {\rho}\frac{e_f}{\varpi}        & ~ \text{if } || e_f || < \varpi,\\
  \end{cases}  \label{input rob}\\
&  {\rho}= Y(\xi)^{T} {\Theta} ,  \label{c1}
\end{align}
where $\Delta \tau$ provides robustness against $\sigma$ through switching gain ${\rho}$; $\varpi \in \mathbb{R}^{+}$ is a small scalar used for chattering removal; $G \in \mathbb{R}^{n \times n}$ is a positive definite matrix. 

\textbf{Evaluation of Switching Gain:}
Evaluation of $\rho$ like (\ref{c1}) is conservative in nature and evidently requires the knowledge of $\Theta^{*}$, which is not always possible in the face of uncertain parametric variations and external disturbances.   
The control laws developed in \cite{Ref:20}-\cite{Ref:22} and \cite{Ref:ut}-\cite{Ref:frid} assume that $\sigma$ and $\dot{\sigma}$ is upper bounded by constant, respectively. Exploring the structure of $||\sigma||$ from (\ref{sigma1}) it can be easily inferred that such constant bound assumption on the uncertainties, whether known or unknwon, puts a restriction on the states a priori. Moreover, the switching gain in \cite{Ref:21}-\cite{Ref:22} suffers from overestimation-underestimation problem.
\subsection{Adaptive Switching-gain based Robust Control (ASRC)}
The major aims of the proposed ASRC framework are:
\begin{itemize}
\item To compensate the uncertainties that can be either LIP or NLIP. However, the upper bound of the uncertainties satisfies the LIP property (\ref{sigma1}).
\item To alleviate the overestimation-underestimation problem of switching gain. 
\end{itemize}
The control input of the proposed ASRC is designed as
\begin{align}
&\tau =-e-G e_f-\Delta \tau, ~ \Delta \tau=\begin{cases}
    \hat{\rho}\frac{e_f}{|| e_f||}       & ~ \text{if } || e_f|| \geq \varpi\\
    \hat{\rho}\frac{e_f}{\varpi}        & ~ \text{if } || e_f || < \varpi,\\
  \end{cases}  \label{input}\\
& \hat{\rho}=\hat{\theta}_0+\hat{\theta}_1||\xi||+ \hat{\theta}_2||\xi||^2+\gamma \triangleq Y(\xi)^T \hat{\Theta}+\gamma,\label{rho}
\end{align}
where $\Delta \tau$ provides robustness against $\sigma$ through $\hat{\rho}$; $\hat{\Theta}=[\hat{\theta}_0 ~ \hat{\theta}_1 ~\hat{\theta}_2]^T$ is the estimate of ${\Theta}$; $\gamma$ is an auxiliary gain. The importance of $\gamma$ will be explained later. The gains $\gamma, \hat{\theta}_i$, $i=0,1,2$ are evaluated using the following adaptive laws: 
\begin{align} 
(i) &~ \text{for}~|| e_f || \geq \varpi \nonumber\\
 \dot{\hat{\theta}}_i=&
  \begin{cases}
   {\alpha}_i ||\xi||^i ||e_f|| ~~\text{if}~\lbrace e^T \dot{e} >0 \rbrace \cup \lbrace \bigcup_{i=0}^{2}\hat{\theta}_i \leq 0 \rbrace \\
  \qquad \qquad \qquad \qquad \qquad \quad \cup \lbrace \gamma \leq \beta \rbrace \\
   -{\alpha}_i ||\xi||^i ||e_f||~~\text{otherwise},
   \end{cases} \label{gain 11} \\
  \dot{\gamma}= &
 \begin{cases}
 \alpha_{3} ||e_f||  ~~\text{if}~\lbrace e^T \dot{e} >0 \rbrace \cup \lbrace \bigcup_{i=0}^{2}\hat{\theta}_i \leq 0 \rbrace \\
 \qquad \qquad \qquad \qquad \qquad \quad \cup \lbrace \gamma \leq \beta \rbrace \\
 -\varsigma \alpha_{3} ||\xi||^{4} ~~\text{otherwise}, \\
 \end{cases} \label{gain 222} \\
(ii)&~\text{for}~ || e_f|| < \varpi \nonumber\\ 
 & \dot{\hat{\theta}}_i=0, \dot{\gamma}=0, \label{gain 22} \\
 \text{with}~ & \hat{\theta}_i(t_0) > 0, i=0,1,2,~ \gamma (t_0)> \beta. \label{init}
\end{align}
Here $t_0$ is the initial time and $\beta,\varsigma, \alpha_0, \alpha_1, \alpha_2, \alpha_3  \in \mathbb{R}^{+}$ are user defined scalars. Substituting (\ref{input}) into (\ref{mr dot}), the closed loop system is formed as:
\begin{align}
M\dot{e}_f&=  -e-G e_f-\Delta \tau -C e_f +\sigma. \label{cl loop}
\end{align}
\begin{remark}
For $||e_f|| \geq \varpi$, it can be noticed from the adaptive laws (\ref{gain 11})-(\ref{gain 222}) that the gains $\hat{\theta}_i, \gamma$ increase if error trajectories move away from $||e||=0$ (governed by $e^T\dot{e}>0$) and decrease if error trajectories do not move away from $||e||=0$ (governed by the 'otherwise' condition in (\ref{gain 11})-(\ref{gain 222}) which implies $\lbrace e^T\dot{e} \leq 0 \rbrace \cap \lbrace \bigcup_{i=0}^{2}\hat{\theta}_i >0 \rbrace \cap \lbrace \gamma>\beta \rbrace$). Hence, the proposed law certainly does not make the switching gain a monotonically increasing function and thus alleviates the overestimation problem.
\end{remark}
\begin{remark}
For $||e_f|| < \varpi$, the tracking error remains bounded inside the ball $B_{\varpi} \triangleq \lbrace||\Gamma \xi|| < \varpi \rbrace$ using the relation $e_f=\Gamma \xi$. This implies that the switching gains are sufficient enough to keep the error within $B_{\varpi}$. Hence, the gains are kept unchanged for $||e_f|| < \varpi$. One can choose small $\varpi$ to improve tracking accuracy (as $B_{\varpi}$ gets reduced) as long as the value of $\varpi$ does not invite chattering.   
\end{remark}
\begin{remark}
The initial condition of the gains are selected as $\hat{\theta}_i(t_0) > 0, ~ \gamma (t_0)> \beta$. Further, for $||e_f|| \geq \varpi$, the adaptive laws (\ref{gain 11})-(\ref{gain 222}) force the gains to increase if either of the gains attempt to breach their respective lower bounds (governed by $\lbrace \bigcup_{i=0}^{2}\hat{\theta}_i \leq 0 \rbrace \cup \lbrace \gamma \leq \beta \rbrace$). This ensures that $\gamma (t) \nless \beta, ~ \hat{\theta}_i(t) \nless 0$ $\forall i=0,1,2$ when $||e_f|| \geq \varpi$. Again, gains remain unchanged for $||e_f|| < \varpi$. Hence, combination of the conditions mentioned above implies
%
\begin{align}
\hat{\theta}_i(t) \geq 0 ~\forall i=0,1,2 ~ \text{and}~ \gamma (t) \geq \beta ~~\forall t\geq t_0. \label{low bound}
\end{align} 
The condition (\ref{low bound}) is later exploited in stability analysis.
\end{remark}
To guarantee the alleviation of the overestimation problem of switching gain, it is necessary that $\hat{\theta}_i, \gamma$ decrease within a finite time. This is shown through Theorem \ref{th finite time }.
\begin{theorem} \label{th finite time }
Let $t=t_{in}$ be any time instant when gains start increasing. Then there exist finite times $t_1,t_2,t_3, \delta t$ such that the gains $\hat{\theta}_0, \hat{\theta}_1, \hat{\theta}_2, \gamma$ decrease for $t\geq t_{in}+T$ where $T \leq \bar{t}+\delta t$, $ \bar{t}=max \lbrace t_1,t_2,t_3\rbrace$. These times are obtained as 
\begin{align}
t_1&\leq \frac{\theta_0^{*}}{(\alpha_0+\alpha_3) \varpi }, t_2 \leq \frac{\theta_1^{*} ||\Gamma|| }{\alpha_1\varpi^2},
t_3 \leq \frac{\theta_2^{*}|| \Gamma || ^2}{\alpha_2\varpi^3}, \label{t}\\
\delta t & \leq (1/ \varrho) \text{ln} \lbrace 2 V (\bar{t})/(||e(\bar{t})||^2) \rbrace ,
\end{align}
\end{theorem}
\noindent where $V=\frac{1}{2}e_f^T M e_f+\frac{1}{2}e^Te$, $\varrho \triangleq \frac{\min\lbrace \lambda_{\min}(G), \lambda_{\min}(\Omega)\rbrace}{ \max \lbrace \mu_2, 1\rbrace} $.
\begin{proof}
Here, $t_{in}$ can be any time when gains start increasing and it is solely used for analysis. The objective of Theorem \ref{th finite time } is to find when the gains start to decrease. Further, it is to be noted from the laws (\ref{gain 11})-(\ref{gain 22}) that the gains increase only when $||e_f|| \geq \varpi $. So, it is sufficient to investigate the condition when all the gains increase and $||e_f|| \geq \varpi$. Moreover, using $e_f= \Gamma \xi$ from (\ref{r}) one has
\begin{align}
&\varpi \leq ||e_f|| \leq || \Gamma|| ||\xi|| \Rightarrow  ||\xi|| \geq (\varpi / || \Gamma||). \label{e}
\end{align}
So, the first laws of (\ref{gain 11}), (\ref{gain 222}) and the condition (\ref{e}) yields
\begin{equation}
\dot{\hat{\theta}}_0 \geq \alpha_0 \varpi,\dot{\hat{\theta}}_1 \geq (\alpha_1 \varpi^2)/ ||\Gamma||, \dot{\hat{\theta}}_2 \geq (\alpha_2 \varpi^3)/ || \Gamma|| ^2, \dot{\gamma} \geq \alpha_3 \varpi. \label{inequal 1}
\end{equation}
Let $V$ be a Lypaunov function. Using (\ref{cl loop}) and the relation $e^T\dot{e}=e^T(e_f-\Omega e)$ (from (\ref{r})), the time derivative of $V$ yields 
\begin{align}
\dot{V}&=e_f^TM\dot{e}_f+(1/2)e_f^T\dot{M}e_f +e^T\dot{e}\nonumber\\
& =e_f^T(-e-G e_f-\Delta \tau +\sigma)+(1/2)e_f^T(\dot{M}-2C)e_f \nonumber\\
& \qquad \qquad \qquad +e^T(e_f-\Omega e). \label{vdot old} 
\end{align}
Further, substituting (\ref{input}) into (\ref{vdot old}) and using Property 1 (implying $e_f^T(\dot{M}-2C)e_f=0$), $\dot{V}$ is simplified as
\begin{align}
\dot{V}& = -e_f^TG e_f -e^T \Omega e+e_f^T(-\hat{\rho}({e_f}/{|| e_f ||}) +\sigma)\nonumber\\
& \leq -e_f^T G e_f -e^T \Omega e-(Y(\xi)^T \hat{\Theta}+ \gamma) ||e_f||+ Y(\xi)^T {\Theta^{*}}  ||e_f||\nonumber\\
& \leq -\lambda_{\min}(G)||e_f||^2-\lambda_{\min}(\Omega)||e||^2 -\lbrace(\hat{\theta}_0+\gamma -\theta_0^{*})\nonumber\\ 
&\qquad \qquad +(\hat{\theta}_1-\theta_1^{*}) ||\xi||+(\hat{\theta}_2 -\theta_2^{*})||\xi||^2\rbrace ||e_f||. \label{vdot old 1}
\end{align} 
Thus, the sufficient condition to achieve $\dot{V} <0$ would be  
\begin{equation}
\hat{\theta}_0+\gamma \geq\theta_0^{*}, ~\hat{\theta}_1\geq \theta_1^{*} ~\text{and} ~\hat{\theta}_2 \geq \theta_2^{*}. \label{inequal 2}
\end{equation}  
Let the system (\ref{sys}) does not posses finite time escape \cite{Ref:hadad}. Then integrating both sides of the inequalities in (\ref{inequal 1}) and using those results in (\ref{inequal 2}) lead to the expressions of $t_1,t_2,t_3$ in (\ref{t}). So, for $t \geq t_{in}+\bar{t}$
\begin{align}
\dot{V}  & \leq -\lambda_{\min}(G)||e_f||^2-\lambda_{\min}(\Omega)||e||^2 \nonumber \\
& \leq -{\varrho}_m  (||e_f||^2+ ||e||^2) ,\label{1}
\end{align}
where ${\varrho}_m \triangleq \min\lbrace \lambda_{\min}(G), \lambda_{\min}(\Omega)\rbrace $. Further, the definition of $V$ yields
\begin{align}
V& \leq  {\varrho}_M (||e_f||^2+ ||e||^2), \label{2}
\end{align}
where ${\varrho}_M \triangleq \max \lbrace \mu_2, 1\rbrace $. Substituting (\ref{2}) into (\ref{1}) and using the comparison Lemma \cite{Ref:khalil} yields
\begin{align}
\dot{V} & \leq - \varrho {V} \Rightarrow
 {V} (t)  \leq {V} (t_{in}+\bar{t})e^{-\varrho (t-\bar{t})} ~\forall t \geq t_{in}+\bar{t},  \label{3}
\end{align}
where $\varrho \triangleq {\varrho}_m / {\varrho}_M$. Here $\hat{\theta}_i>0, \gamma >\beta $ as gains were increasing. So, to ensure the `\textit{otherwise}' condition (i.e. $\lbrace e^T\dot{e} \leq 0 \rbrace \cap \lbrace \bigcup_{i=0}^{2}\hat{\theta}_i >0 \rbrace \cap \lbrace \gamma>\beta \rbrace$), the condition $e^T\dot{e} \leq 0$ (i.e. $||e (t)|| $ does not increase) needs to take place. From the definition of $V$, the upper bound of $e$ follows
\begin{align}
V(t) \geq (1/2)||e (t)||^2 \Rightarrow ||e (t)|| \leq \sqrt{2V(t)}~ ~\forall t \geq t_0. \label{4}
\end{align}
Let $||e(t_{in}+\bar{t})||=\psi$ which implies $ V(t_{in}+\bar{t}) \geq \psi^2 /2$ from (\ref{4}). As $V(t)$ decreases exponentially $\forall {t \geq t_{in}+\bar{t}}$ following (\ref{3}), there exist a finite time $\delta t=t-(t_{in}+\bar{t})$ such that $V(t_{in}+\bar{t}+\delta t) = \psi^2 /2$ implying $||e (t_{in}+\bar{t}+\delta t)|| \leq \psi$. So, $\lbrace e^T\dot{e} \leq 0 \rbrace \cap \lbrace \bigcup_{i=0}^{2}\hat{\theta}_i >0 \rbrace \cap \lbrace \gamma>\beta \rbrace$ would occur at $t\geq t_{in}+T$ where $T \leq \bar{t}+ \delta t$ and $\hat{\theta}_i, \gamma$ would start decreasing. The time $\delta t$ is found from (\ref{3}):
\begin{align}
 &\psi^2 \leq 2V(t_{in}+\bar{t})e^{-\varrho \delta t}, ~  \quad \forall t \geq t_{in}+\bar{t} \nonumber\\
\Rightarrow \delta t & \leq (1/ \varrho) \text{ln} \lbrace 2 V(t_{in}+\bar{t})/\psi^2 \rbrace.
\end{align}
\end{proof}
\begin{remark}
The increment and decrement of $\hat{\theta}_i, \gamma$ can occur several times depending on the error incurred by the system. However, time interval $\Delta t$ between two successive decrement will always satisfy $\Delta t \leq \bar{t}+\delta t$. Moreover, high values of $\alpha_0, \alpha_1, \alpha_2, \alpha_3 $ help to reduce $\bar{t}$ and achieve faster adaptation.  
\end{remark}
\subsection{Stability Analysis}
Exploring the structures of the adaptive laws (\ref{gain 11})-(\ref{gain 22}), three possible scenarios are identified: Case (1): $\dot{\hat{\theta}}_i, \dot{\gamma}$ increase and $|| e_f || \geq \varpi$; Case (2): $\dot{\hat{\theta}}_i, \dot{\gamma}$ decrease and $|| e_f || \geq \varpi$; Case (3) $\dot{\hat{\theta}}_i=0, \dot{\gamma}=0$ when $|| e_f || < \varpi$ $\forall i=0,1,2$. 
\begin{theorem}\label{th 2}
The closed loop system (\ref{cl loop}) with control input (\ref{input})-(\ref{gain 22}) guarantees $e(t),e_f(t), \tilde{\theta}_i(t), \gamma (t)$ to be Uniformly Ultimately Bounded (UUB) where $\tilde{\theta}_i \triangleq (\hat{\theta}_i-{\theta}_i^{*})$, $i=0,1,2$. 
\end{theorem}
\begin{proof}
%
The stability analysis of the overall system is carried out for the three cases mentioned above using the following common
Lyapunov function:
\begin{equation}
V_1=V+\sum_{i=0}^{2}\frac{1}{2\alpha}_i \tilde{\theta}_i ^2+ \frac{1}{2\alpha_{3}} \gamma^2.\label{lya el}
\end{equation}
\textbf{Case (1)}: $\dot{\hat{\theta}}_i,\dot{\gamma}$ increase $\forall i=0,1,2$ and $|| e_f || \geq \varpi$.\\
Note that $\sum_{i=0}^{2}\frac{1}{\alpha_i}\tilde{\theta}_i  \dot{\hat{\theta}}_i= Y(\xi)^T (\hat{\Theta}-{\Theta}^{*}) ||e_f||$. Then using (\ref{gain 11})-(\ref{gain 222}) and following the procedure in (\ref{vdot old 1}) one obtains 
\begin{align}
\dot{V}_1& \leq -e_f^TG e_f -e^T \Omega e+e_f^T(-\hat{\rho}({e_f}/{|| e_f ||}) +\sigma)\nonumber\\
& \qquad \qquad \qquad \qquad +Y(\xi)^T(\hat{\Theta}-{\Theta}^{*}) ||e_f||+\gamma ||e_f||\nonumber\\
& \leq -e_f^T G e_f -e^T \Omega e-(Y(\xi)^T \hat{\Theta}+\gamma) ||e_f||+ Y(\xi)^T {\Theta^{*}}  ||e_f||\nonumber\\
& \qquad \qquad \qquad \qquad \qquad +Y(\xi)^T(\hat{\Theta}-{\Theta}^{*}) ||e_f|| +\gamma ||e_f|| \nonumber\\
& \leq -\lambda_{\min}(G)||e_f||^2-\lambda_{\min}(\Omega)||e||^2 \leq 0. \label{vdot new1}
\end{align}
From (\ref{vdot new1}) it can be inferred that ${V}_1(t) \in \mathcal{L}_{\infty}$. Thus, the definition of $V_1$ yields $\gamma(t), \tilde{\theta}_i(t), e(t), e_f(t) \in \mathcal{L}_{\infty} \Rightarrow \xi(t),\hat{\theta}_i  \in \mathcal{L}_{\infty}$. 
Thus, for this case the closed-loop system remains stable. 

The gains $\gamma, \hat{\theta}_i$, $i=0,1,2$ remain bounded for Case (1), decrease for Case (2) and remain constant for Case (3). Hence,
$\exists$~$\bar{\gamma}, \bar{\theta}_i \in \mathbb{R}^{+}$ such that
\begin{align}
&~ \hat{\theta}_0 (t) \leq \bar{\theta}_0, \hat{\theta}_1 (t) \leq \bar{\theta}_1, \hat{\theta}_2 (t) \leq \bar{\theta}_2, \gamma (t) \leq \bar{\gamma} ~ \forall t\geq t_0.\label{up bound}
\end{align} 
\textbf{Case (2)}: $\dot{\hat{\theta}}_i, \dot{\gamma}$ decrease $\forall i=0,1,2$ and $|| e_f|| \geq \varpi$.\\
Using $\gamma \geq \beta$ (from (\ref{low bound})) and $||{e_f}|| \leq ||\Gamma|| ||\xi||$ yields: 
\begin{align}
\dot{V}_1& \leq -e_f^T G e_f -e^T \Omega e+ e_f^T(-\hat{\rho} ({e_f}/{|| e_f||}) +\sigma)\nonumber\\
& \qquad \qquad \qquad \qquad -Y(\xi)^T (\hat{\Theta}-{\Theta}^{*}) ||e_f||-\gamma \varsigma ||\xi||^4 \nonumber\\
& \leq -e_f^T G e_f -e^T \Omega e-(Y(\xi)^T \hat{\Theta}+\gamma) ||e_f||+ Y(\xi)^T {\Theta^{*}} ||e_f||\nonumber\\
& \qquad \qquad \qquad \qquad-\varsigma \beta ||\xi||^4-Y(\xi)^T(\hat{\Theta}-{\Theta}^{*}) || e_f || \nonumber\\
&\leq -\varsigma \beta ||\xi||^4 +2||\Gamma|| \lbrace \theta_0^{*}+ \theta_1^{*}||\xi|| + \theta_2^{*}||\xi||^2\rbrace  ||\xi|| \nonumber\\
& \qquad \qquad  -\lambda_{\min}(G)|| e_f ||^2-\lambda_{\min}(\Omega)||e||^2. \label{vdot new2}
\end{align}
Since $0 \leq \hat{\theta}_i (t) \leq \bar{\theta}_i$, $\beta \leq \gamma \leq \bar{\gamma}$ (from (\ref{low bound}) and  (\ref{up bound})), the definition of $V_1$ in (\ref{lya el}) yields
\begin{align}
V_1 \leq \varrho_M (||e_f||^2+||e||^2)+\zeta, \label{v bound 1}
\end{align}
where $\zeta \triangleq \sum_{i=0}^{2}\frac{1}{\alpha}_i  ({{\theta}_i^{*}}^2+\bar{\theta}_i^2)+ \frac{1}{\alpha_{3}} \bar{\gamma}^2$. Thus using (\ref{v bound 1})
\begin{align}
-\lambda_{\min}(G)||e_f||^2-\lambda_{\min}(\Omega)||e||^2  \leq -\varrho V_1 + \varrho \zeta. \label{v bound 11}
\end{align}
 Substitution of (\ref{v bound 11}) into (\ref{vdot new2}) yields
\begin{align}
\dot{V}_1& \leq -\varrho V_1+ f_{p}(||\xi||), \label{v bound 12}
\end{align}
{  where $f_{p}(||\xi||)=-\varsigma \beta ||\xi||^4 +2||\Gamma|| \lbrace \theta_0^{*}||\xi||+ \theta_1^{*}||\xi||^2 + \theta_2^{*}||\xi||^3 \rbrace+ \varrho \zeta$. Applying Descarte's rule of sign change \cite{Ref:des}, one can verify that $f_{p}(||\xi||)$ has maximum one positive real root. Further, it is to be noticed that $f_{p}(|| \xi ||=0)=\varrho \zeta \in \mathbb{R}^{+}$ and  $f_{p}(||\xi||) \rightarrow  -\infty$ as $||\xi|| \rightarrow \infty$. Hence, according to Bolzano's Intermediate Value Theorem \cite{Ref:bolz}, $f_{p}(||\xi||) $ will have at least one positive real root. So, combination of the Intermediate Value Theorem and Descarte's rule of sign change reveals that $f_{p}(||\xi||) $ has exactly one positive real root. Therefore, the nature of roots of $f_{p}(||\xi||) $  will be either (i) one positive real root and three negative real roots or (ii) one positive real root, one negative real root and a pair of complex conjugate roots. 

Let $\iota \in \mathbb{R}^{+}$ be the positive real root of $f_{p}(||\xi||)$. Figure \ref{fig:root} depicts the nature of $f_{p}(||\xi||)$ depending on the various combination of roots \cite{Ref:poly book}. It is to be noted that the actual graph and values of the roots of the polynomial $f_{p}(||\xi||)$ would depend on the values of the coefficients of $f_{p}(||\xi||)$. However, $\theta_0^{*}, \theta_1^{*}, \theta_2^{*}$ and $\zeta$ are unknown here. Nevertheless, to study the stability of the system, it is sufficient to analyse the nature of $f_{p}(||\xi||)$ (i.e. the instances when $f_{p}(||\xi||)>0$ or $f_{p}(||\xi||) \leq 0$) rather than determining the values of the roots. The nature of any polynomial can be understood from the occurrence of its real roots \cite{Ref:poly book}. Moreover, the leading coefficient of $f_{p}(||\xi||)$ (the coefficient of the highest degree term $||\xi||^4$) is negative (as $\varsigma, \beta \in \mathbb{R}^{+}$). As a matter of fact, one can notice from Fig. \ref{fig:root} that $f_{p}(||\xi||) \leq 0$ when  $||\xi|| \geq \iota$. Hence, the overall system would be UUB for this case. From the point of view of controller design, it is important to reduce $\iota$ such that better tracking accuracy can be achieved and this can be obtained by increasing $\varsigma$.

\begin{figure}[h]
\begin{center}
\includegraphics[width=3.5 in,height=2.3 in]{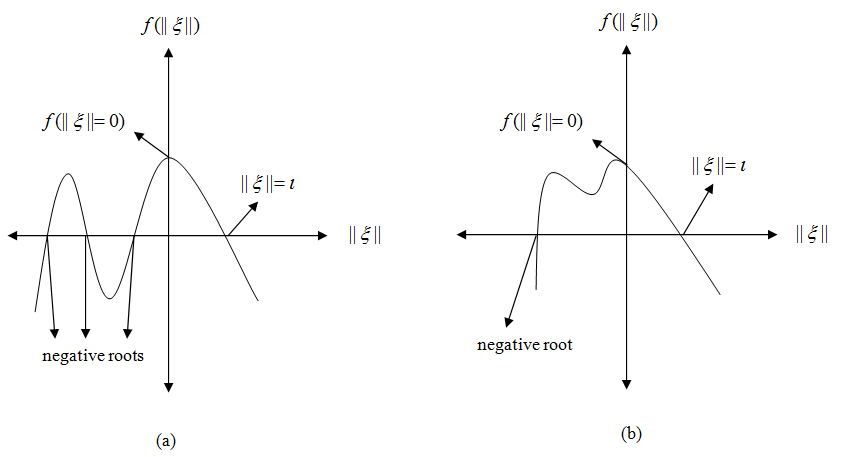}
\caption{ {  $f_{p}(||\xi||)$ with (a) one positive real root and three negative real roots (b) one positive real root, one negative real root and a pair of complex conjugate roots.}}\label{fig:root}
\end{center}
\end{figure}
}
Hence, the overall system would be UUB \cite{Ref:khalil} for this case. 

\textbf{Case 3}: $\dot{\hat{\theta}}_i=0, \dot{\gamma}=0$, $\forall i=0,1,2$ when $|| e_f || < \varpi$. \\
Similar to Case (1), $\dot{V}_1$ can be simplified for this case as
\begin{align}
\dot{V}_1& \leq -e_f^TG e_f -e^T \Omega e+e_f^T(-\hat{\rho}({e_f}/{\varpi}) +\sigma)\nonumber\\
& \leq -e_f^TG e_f -e^T \Omega e-(Y(\xi)^T \hat{\Theta} + \gamma)({||e_f||^2}/{\varpi})\nonumber\\
& \qquad \qquad + Y(\xi)^T {\Theta^{*}} ||e_f|| \nonumber\\
& \leq -\lambda_{\min}(G)||e_f||^2-\lambda_{\min}(\Omega)||e||^2 + Y(\xi)^T {\Theta^{*}}  ||e_f||. \label{vdot new1 sup}
\end{align}
For $||e_f|| < \varpi$, the system remains bounded inside the ball $B_{\varpi} \triangleq \lbrace||\Gamma \xi|| < \varpi \rbrace$ as $e_f=\Gamma \xi$. This implies that $Y(||\xi||) \in \mathcal{L}_{\infty}$. Hence, for $||e_f|| < \varpi$, $\exists \vartheta \in \mathbb{R}^{+}$ such that 
\begin{equation}
||  Y(\xi)^T {\Theta^{*}} || || e_f || \leq \varpi \vartheta. \label{sup 31}
\end{equation}
Let us define a scalar $z$ as $0 < z < \lambda_{\min}(G)$. Then using (\ref{v bound 1}) and (\ref{sup 31}),  (\ref{vdot new1 sup}) is modified as
\begin{align}
\dot{V}_1& \leq -\lambda_{\min}(\Omega)||e||^2- \lbrace \lambda_{\min}(G)-z \rbrace || e_f ||^2 - z|| e_f ||^2+ \varpi \vartheta\nonumber\\
       & \leq  - \varrho_{1} V_1- z|| e_f ||^2+ \varrho_{1} \zeta +  \varpi \vartheta, \label{sup 32}
\end{align}
where $\varrho_{1} \triangleq \lbrace \min \lbrace ( \lambda_{\min}(G)-z), \lambda_{\min}(\Omega)\rbrace \rbrace / \varrho_M$. Hence, the system would be UUB when 
\begin{equation}
|| e_f ||=||\Gamma \xi || \geq \sqrt{({\varrho_{1} \zeta +  \varpi \vartheta})/{z}}.
\end{equation}
Since the closed loop system remains UUB for both the cases $||e_f || \geq \varpi$ and $||e_f|| <\varpi$ using the common Lyapunov function (\ref{lya el}), the overall closed loop system also remains UUB \cite{Ref:liber}.
\end{proof}
\begin{remark}
It is noteworthy that the condition (\ref{low bound}) is necessary for stability of the system. Moreover, high values of $\varsigma$ helps to reduce $\iota$ which consequently can improve controller accuracy. However, one needs to be careful that too high value of $\varsigma$ may excite the condition $\gamma \leq \beta$ leading to the increment in all the gains $\gamma, \hat{\theta}_i,$ $i=0,1,2$. Further, the scalar terms $z, \vartheta, \psi,\mu_2,\zeta,\bar{\theta}_i$ and $\bar{\gamma}$ are only used for the purpose of analysis and not used to design control law.
\end{remark}
\begin{remark} \label{poly}
The importance of the auxiliary gain $\gamma$ can be realized from Theorems 1 and 2. It can be observed from (\ref{t}) that $t_1$ gets reduced due to the presence of $\alpha_3$ (contributed by $\dot{\gamma}>0$) which leads to faster adaptation. Moreover, the negative fourth degree term $-\varsigma \beta ||\xi||^4$ in $f_p(||\xi||)$ (contributed by $\dot{\gamma}<0$) ensures system stability for Case (2) by making $f_p(||\xi||) \leq 0$ for $|| \xi || \geq \iota $. This also indicates the reason for selecting $\beta>0$ while lower bounds of other gains $\hat{\theta}_i$ $i=0,1,2$ are selected as zero. 
\end{remark}

\textbf{Special Case:} 
The quadratic term $||\xi||^2$ in uncertainty bound (\ref{sigma1}) is contributed by the matrix $C(q,\dot{q})$ (through property 4 in (\ref{i})). EL systems such as robotic manipulator, underwater vehicles, ship dynamics etc. includes $C(q,\dot{q})$. However, there also exist second order EL system (e.g. reduced order WMR system) which does not have the term $C(q,\dot{q})$. For such systems, the following LIP structure would hold:
\begin{equation}
||\sigma ||\leq \theta_0^{*}+\theta_1^{*}||\xi|| \triangleq Y(\xi)^T \Theta^{*}, \label{sigma11} 
\end{equation}
where $Y(\xi)=[1 \quad ||\xi|| ]^T$ and $\Theta^{*}=[\theta_0^{*} \quad \theta_1^{*}]^T$. Hence, following the switching gain laws (\ref{rho})-(\ref{gain 22}), the control laws for uncertainty structure (\ref{sigma11}) are modified as 
\begin{align} 
& \hat{\rho}=\hat{\theta}_0+\hat{\theta}_1||\xi||+ \gamma \triangleq Y(\xi)^T \hat{\Theta}+\gamma, \label{rho1}\\
(i) &~ \text{for}~|| e_f|| \geq \varpi \nonumber \\
& \dot{\hat{\theta}}_i=
  \begin{cases}
   {\alpha}_i ||\xi||^i ||e_f|| ~~\text{if}~\lbrace e^T \dot{e} >0 \rbrace \cup \lbrace \bigcup_{i=0}^{1}\hat{\theta}_i \leq 0 \rbrace \\
  \qquad \qquad \qquad \qquad \qquad \quad \cup \lbrace \gamma \leq \beta \rbrace \\
   -{\alpha}_i ||\xi||^i ||e_f||~~\text{otherwise},
   \end{cases} \label{gain 13} \\
 & \dot{\gamma}= 
 \begin{cases}
 \alpha_{3} ||e_f||  ~~\text{if}~\lbrace e^T \dot{e} >0 \rbrace \cup \lbrace \bigcup_{i=0}^{1}\hat{\theta}_i \leq 0 \rbrace \\
 \qquad \qquad \qquad \qquad \qquad \quad \cup \lbrace \gamma \leq \beta \rbrace \\
 -\varsigma \alpha_{3} ||\xi||^{3} ~~\text{otherwise},
 \end{cases} \label{gain 223} \\
 (ii) &~ \text{for}~|| e_f || < \varpi \nonumber \\
 & \dot{\hat{\theta}}_i=0, \dot{\gamma}=0, \label{gain 23} \\
 \text{with}~& \hat{\theta}_i(t_0) > 0, i=0,1,~ \gamma (t_0)> \beta. \label{init 1} 
\end{align}
System stability employing (\ref{rho1})-(\ref{gain 23}) can be analysed exactly like Theorem \ref{th 2} using the following Lyapunov function: 
\begin{equation}
V_1=V+\sum_{i=0}^{1}({1}/{2\alpha}_i) \tilde{\theta}_i ^2+ ({1}/{2\alpha_{3}})\gamma ^2.\label{lya el 1}
\end{equation}
One can verify that the cubic polynomial $2||\Gamma|| \lbrace \theta_0^{*}+ \theta_1^{*}||\xi|| + \theta_2^{*}||\xi||^2\rbrace ||\xi||$ in $f_p (||\xi||)$ of Case (2) would be modified as quadratic polynomial $2||\Gamma|| \lbrace \theta_0^{*}+ \theta_1^{*}||\xi|| \rbrace ||\xi||$ using (\ref{sigma11}) and (\ref{lya el 1}). Hence, following the argument in Remark \ref{poly}, it can be noticed that a cubic term $ -\varsigma \alpha_{3} ||\xi||^{3}$ is selected in the adaptive law (\ref{gain 23}) for system stability.

Thus, with EL system (\ref{sys}), only two structures are possible for $||\sigma||$: (i) $Y(\xi)=[1 ~ ||\xi||~||\xi||^2 ]^T$, $\Theta^{*}=[\theta_0^{*} ~ \theta_1^{*}~ \theta_2^{*}]^T$ and (ii) $Y(\xi)=[1 ~ ||\xi|| ]^T$, $\Theta^{*}=[\theta_0^{*} ~\theta_1^{*}]^T$. Both these situations are covered here. 
For better inference, the ASRC algorithm is summarized in Table \ref{table algo} for various system structures.
\begin{table}[!h]
	\renewcommand{\arraystretch}{1.3}
	\caption{ASRC Algorithm for Various System Structures}
	\label{table algo}
	\centering
	\begin{tabular}{c c c c}
		\hline
		\hline
		\multicolumn{2}{c}{System Structure} & LIP structure of $|| \sigma || $ & Control law   \\ \cline{2-4}
		\hline
		\multirow{2}{*}{(\ref{sys})} & $C(q,\dot{q}) \neq 0$ & (\ref{sigma1})  & (\ref{input}) - (\ref{init})\\ 
		& $C(q,\dot{q})=0$ & (\ref{sigma11}) &  (\ref{input}), (\ref{rho1}) - (\ref{init 1})\\ 
		\hline
		\hline
	\end{tabular}
\end{table}

\textbf{Comparison with existing Adaptive-Robust law}: To gain further insight into the advantage of the proposed adaptive law, the following adaptive law of ASMC \cite{Ref:21}-\cite{Ref:22} for switching gain $K$ is provided:
\begin{align}
& \dot{K}(t)=
  \begin{cases}
   \bar{K}||s||\text{sgn}(||s||-\epsilon), ~\text{if}~K > \beta \\
  \beta  \qquad \qquad \qquad \qquad ~\text{if}~K \leq \beta,
   \end{cases} \label{asmc 1}
\end{align} 
where $\epsilon, \bar{K} \in \mathbb{R}^{+}$ are user defined scalars and $s$ is the sliding surface. It can be observed from (\ref{asmc 1}) that when $||s|| \geq \epsilon$ the switching gain $K$ increases monotonically even if error trajectories move close to $||s|| = 0$. This gives rise to the overestimation problem of switching gain. Again, even if $K$ is sufficient to keep $||s||$ within $\epsilon$, it decreases monotonically when $||s|| < \epsilon $. So, at certain time, $K$ would become insufficient and error will increase again. However, $K$ will not increase (rather it keeps on decreasing) until $||s|| >\epsilon$, which creates underestimation problem. Low (resp. High) value of $\epsilon$ may force $K$ to increase (resp. decrease) for longer duration when $||s|| \geq \epsilon$ (resp. $||s|| < \epsilon$) resulting in escalation of the overestimation (resp. underestimation) problem of ASMC.

Whereas, ASRC allows its gains to decrease when error trajectories move towards $||e||=0$ and $||e_f|| \geq \varpi$ (overcoming overestimation problem) and keeps the gains unchanged when they are sufficient to keep the error within the ball $B_\varpi$ (overcoming underestimation problem). Since the overestimation-underestimation problems are alleviated by ASRC for any $\varpi$,  one can in fact reduce $\varpi$ for better tracking accuracy as long as chattering does not occur. 

\section{Application: Nonholonomic WMR}
\begin{figure}[h]
\begin{center}
\includegraphics[width=3 in,height=2.5 in]{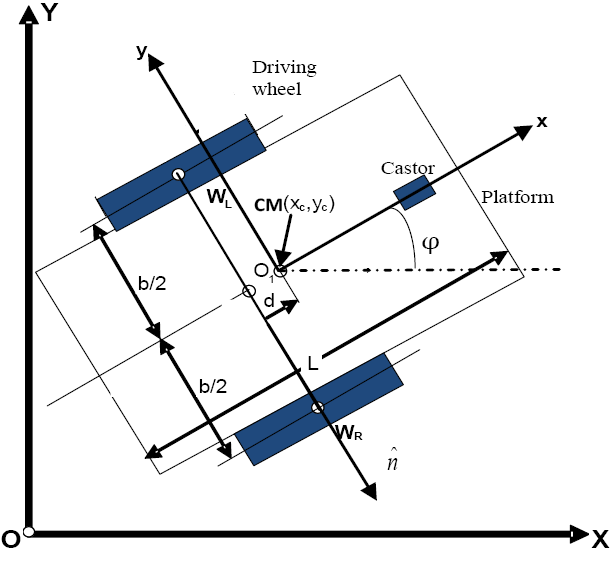}
\caption{Schematic of a WMR.}\label{fig:wmr}
\end{center}
\end{figure}
Nonholonomic WMR, which has vast applications in transportation, planetary exploration, surveillance, security, human-machine-interfaces etc., provides a unique platform to test the proposed control law.  
Hence, the performance of the proposed ASRC is verified using a commercially available 'PIONEER 3' WMR in comparison to Adaptive Sliding Mode Control (ASMC) \cite{Ref:21}-\cite{Ref:22}. The ASMC law is detailed in \cite{Ref:21}-\cite{Ref:22} while it follows the adaptive law (\ref{asmc 1}). 

The Euler-Lagrangian formulation of a nonholonomic WMR (Fig. \ref{fig:wmr}) is given as \cite{Ref:coelho}: 
\begin{align}
&{M}(q)\ddot{q}+ {C}(q,\dot{q})\dot{q}=L\tau-{A^{*}}^T \lambda^{*}, \label{sys w1}
\end{align}
\begin{align}
\text{where}~&{M}=\begin{bmatrix}
m & 0 & mdsin\varphi & 0 & 0 \\ 
0 & m & -mdcos\varphi & 0 & 0\\ 
mdsin\varphi & -mdcos\varphi & \bar{I} & 0 & 0 \\ 
0 & 0 & 0 & I_w & 0 \\ 
0 & 0 & 0 & 0 & I_w
\end{bmatrix}, \nonumber\\
L=&\begin{bmatrix}
0 &0 \\ 
0 &0 \\ 
0 & 0\\ 
1 & 0\\ 
0 & 1
\end{bmatrix},
{C}(q,\dot{q})\dot{q}=\begin{bmatrix}
md \dot{\varphi}^2 cos \varphi\\ 
md \dot{\varphi}^2 sin \varphi\\ 
0\\ 
0\\ 
0\\ 
\end{bmatrix}, \tau=\begin{bmatrix}
\tau_r\\ 
\tau_l
\end{bmatrix}. \nonumber
\end{align}
Here $q \in\mathbb{R}^5=\lbrace x_c,y_c, \varphi, \theta_r, \theta_l \rbrace$ is the generalized coordinate vector of the system; $(x_c,y_c)$ are the coordinates of the center of mass (CM) of the system and $\varphi$ is the heading angle; $(\theta_r, \theta_l)$ and $(\tau_r,\tau_l)$  are rotation and torque inputs of the right and left wheels respectively; $m,\bar{I}, I_w,r_w$ and $b$ represent the system mass, system inertia, wheel inertia, wheel radius and robot width respectively; $d$ is the distance to the CM from the center of the line joining the two wheel axis; ${A^{*}}$ and $\lambda^{*}$ represent the constraint matrix and vector of constraint forces (Lagrange multipliers) respectively. Expressions of ${A^{*}}$ and $ \bar{I}$ can be found in \cite{Ref:coelho}. 

It is noteworthy that the system (\ref{sys w1}) has only two control inputs although having five generalized coordinates. In fact, for WMR, one can only directly control wheel positions $(\theta_r, \theta_l)$ rather than $(x_c,y_c, \varphi)$. So, the system dynamics can be represented as a combination of a reduced order dynamics and kinematic model for efficient controller design as in \cite{Ref:sir}, \cite{Ref:coelho}:
\begin{align}
&M_R\ddot{q}_R+C_R \dot{q}_R=\tau, \label{sys w4}
\end{align}
\begin{align}
&\dot{q}=\underbrace{\begin{bmatrix}
\frac{r_w}{b}\left ( \frac{b}{2} cos(\varphi)-d sin(\varphi)\right ) & \frac{r_w}{b}\left ( \frac{b}{2} cos(\varphi)+d sin(\varphi)\right )\\ 
\frac{r_w}{b}\left ( \frac{b}{2} sin(\varphi)+d cos(\varphi)\right ) & \frac{r_w}{b}\left ( \frac{b}{2} sin(\varphi)-d cos(\varphi)\right )\\ 
r_w/b & -r_w/b\\
1 & 0\\
0 & 1
\end{bmatrix}}_{S(q)} \dot{q}_R,  \label{S}
\end{align}
\begin{align}
&\text{where}~ M_R=S^TMS=\begin{bmatrix}
k_1 & k_2 \\ 
k_2 & k_1
\end{bmatrix},\label{new m}\\
&k_1= I_w + \lbrace \bar{I}+m (b^2/4-d^2) \rbrace ({r_w}^2/b^2), \nonumber \\
&k_2= \lbrace m (b^2/4+d^2)-\bar{I} \rbrace ({r_w}^2/b^2),\nonumber\\
&C_R=S^T(M \dot{S}+C S)=\begin{bmatrix}
0  & 0 \\ 
0 & 0
\end{bmatrix}, q_R=[\theta_r ~ \theta_l]^T. \label{express}
\end{align}
As WMR moves on ground, the gravity vector $g(q)$ and the potential function would certainly be zero which implies that $M_R, C_R$ satisfies the Properties 1 and 3 \cite{Ref:sir}, \cite{Ref:d'andrea}. The main implication of system Property 1 is to hold $e_f^T(\dot{M}-2C)e_f=0$ and this can be easily verified from (\ref{new m})-(\ref{express}). 
The WMR dynamics (\ref{sys w1}) is based on rolling without slipping condition and hence the term $f(\dot{q}_R)$ is omitted. However, in practical circumstances a WMR is always subjected to uncertainties like friction, slip, skid, external disturbance etc. Hence, incorporating (\ref{express}), the system dynamics (\ref{sys w4}) is modified as
\begin{align}
&M_R\ddot{q}_R+ f(\dot{{q}}_R) + d_s=\tau, \label{sys w2}
\end{align}
where $f(\dot{q}_R) $ and $d_s$ are considered to be the unmodelled dynamics and disturbance respectively. { Often, simple controllers such as open loop control (OLC), PID controller are used in practice for their simplicity. However, the various works such as [26] and the references of [26] (e.g., reference [15] in [26]) have discussed the need to formulate advanced robust tracking controllers for WMR compared to conventional open loop control or PID control to improve tracking accuracy, specifically in the face of unmodeled dynamics and time-varying uncertainty. During the experiment, the payload of the system may be varied due to addition or removal of sensors according to the application requirement; this causes variations in overall system mass, center of mass, inertia etc. Further, the original systems dynamics (\ref{sys w4}) is formed based on the pure rolling assumption. This assumption is not satisfied in practice due to the friction effect between wheel and surface; this is denoted by $f(\dot{q}_R)$ in the WMR dynamics (57). Apart from this, there are also effects of external disturbances $d_s$. 
However, the evaluation of switching gain for robust controller like \cite{Ref:self} requires prior knowledge of the bound of the uncertainties. This implies the designer needs to have the knowledge of $\theta_0^{*},\theta_1^{*}$ for WMR (due to the absence of Coriolis term in the WMR dynamics). This further means that the designer should have the knowledge of the parametric variations in systems as well as bound of $f(\dot{q}_R)$ and $d_s$. This demands tedious modelling job which is also not always accurate.

The benefit, applicability and efficacy of the proposed ASRC can be realized in this context. ASRC does not require any knowledge of the systems dynamics terms $M_R$, $f(\dot{q}_R)$ and $d_s$ of WMR system (57) (and for the matter of fact any EL system representing dynamics (3)). Further, while implementing the control law, it does not need any knowledge of $\theta_0^{*},\theta_1^{*}$ and rather adapts these terms by the adaptive law (46)-(49) (since Coriolis component is zero, the ASRC algorithm applied to WMR is based on the control laws (\ref{input}), (\ref{rho1})-(\ref{init 1})). Hence, ASRC eliminates any effort to model the system as well as avoids any need to characterize the time-varying uncertain parameters and disturbances. }
 It is to be noted that $S(q)$ is only used for coordinate transformation and WMR pose $(x_c, y_c, \varphi)$ representation and, not for control law design.

Hence, the objective is to apply ASRC and ASMC to the reduced order WMR system (\ref{sys w2}) to track a desired $q_R^d(t)$ which in effect track a desired ${q}^d(t)$ through (\ref{S}). To illustrate the fact: one can direct a WMR to move in a specified circular path by designing two suitable different and fixed wheel velocities or in a Lawn-Mower path by applying approximated square-wave velocity profile to the wheels \cite{Ref:self}.

\subsection{Experimental Scenario} 
The WMR is directed to follow a circular path using the following desired trajectories: 
\begin{align*}
\theta_r^d=(4t+ \pi/10) rad, ~ \theta_l^d=(3t+\pi/10) rad. 
\end{align*}
PIONEER 3 uses two quadrature incremental encoders ($500$ ppr) and always starts from $\theta_r(t_0)=\theta_l(t_0)=0$ and the initial wheel position error $(\pi/10, \pi/10)$ rad helps to realize the error convergence ability of the controllers. The desired WMR pose $(x_c^d, y_c^d, \varphi^d)$ and actual WMR pose $(x_c, y_c, \varphi)$ can be determined from (\ref{S}) using $(\dot{\theta}_r^d$, $\dot{\theta}_l^d)$ and $(\dot{\theta}_r$, $\dot{\theta}_l)$ (obtained from encoder) respectively with $r_w=0.097 m, b=0.381 m, d=0.02m$ (supplied by the manufacturer). The control laws for both ASRC and ASMC are written in VC++ environment. Considering the hardware response time, the sampling interval is selected as $20ms$ for all the controllers. Further, to create a dynamic payload variation, a $3.5 kg$ payload is added (kept for $5$ sec) and removed (for $5$ sec) periodically on the robotic platform at different places.  

The controller parameters for ASRC are selected as: $G=\Omega=I$, $\varpi=0.5$, $\hat{\theta}_i(t_0)=\gamma(t_0)=20, \alpha_i=\alpha_3=10$ $\forall i=0,1$, $\beta=0.1,\varsigma=10$. Further, the controller parameters for ASMC are selected as $s= e_f$, $\bar{K}=10,K(t_0)=35,\epsilon=0.5$.

\subsection{Experimental Results and Comparison} 
 The path tracking performance of ASRC is depicted in Fig. \ref{fig:1} while following the desired circular path. The tracking performance comparison of ASRC with ASMC is illustrated in Fig. \ref{fig:2} in terms of $E_p$ ( defined by the Euclidean distance in $x_c,y_c$ error) and $E_{\tau}$ (defined as $||\tau||$). ASMC framework is built on the assumption that uncertainties are upper bounded by an unknown constant (i.e. $\theta_1^{*}=\theta_2^{*}=0$ for general EL systems and $\theta_1^{*}=0$ for this particular WMR based experiment as $C_R=0$). This assumption is restrictive in nature for EL systems and the switching gain is thus insufficient to provide the necessary robustness. As a matter of fact, ASRC provides better tracking accuracy over ASMC. 
 
 To evaluate the benefit of the proposed adaptive-robust law, the evaluation of switching gain for ASMC and ASRC are provided in Fig. \ref{fig:3} and \ref{fig:4} respectively. 
Figure \ref{fig:3} reveals that $K$, the switching gain of ASMC, increases even when $||s||$ approaches towards $||s||=0$ during the time $t$$=$$0$$-$$1.2$ sec. This is due to the fact that $K$ does not decrease unless $||s|| < \epsilon$ and this gives rise to the \textit{overestimation} problem. On the other hand for ASRC, it can be seen from Fig. \ref{fig:4} that all the gains (i.e. $\gamma, \hat{\theta}_0, \hat{\theta}_1$) decrease when $||e_R||$ ($e_R=q_R-q^d_R$) decreases during $t$$=$$0$$-$$1$ sec when $||e_{f_R}|| \geq \varpi$ ($||e_{f_R}||=\dot{e}_R+ \Omega e_R$). So, ASRC overcomes the overestimation problem which is encountered in ASMC. 
Further, $K$ decreases monotonically for time durations $t$$=$$1.2$$-$$38.5$ sec, when $||s|| < \epsilon$. This monotonic decrement makes $K$ insufficient to tackle uncertainties at certain time creating \textit{underestimation} problem. As a result, $||s||$ starts increasing again for $t> 38.5$ sec leading to poor tracking accuracy and $K$ increases again when $||s|| \geq \epsilon$. Gains of ASRC, on the contrary, stay unchanged for $t > 1$~sec when the gains are sufficient to keep $||e_{f_R}|| < \varpi $ avoiding any underestimation problem. {Further the evaluation $\hat{\rho}$ of ASRC is shown in Fig. \ref{fig:5}. It is to be noted that $\hat{\rho}=\hat{\theta}_0+\hat{\theta}_1||\xi||+ \gamma$ for WMR. Hence, though $\hat{\theta}_0, \hat{\theta}_1$ and $\gamma$ remain constant for $t > 1$ sec, $\hat{\rho}$ is not constant for $t > 1$ due to the presence of $\xi$. This is shown by magnifying $\hat{\rho}$ for time durations $t=1-5$ sec and $t=20-25$ sec in Fig. \ref{fig:5}.} While reduction of $\epsilon$ would cause more overestimation problem for ASMC, Table \ref{table 3} shows how tracking accuracy of ASRC improves with reduced $\varpi$ (other control parameters are kept unchanged) while chattering is not observed in control input.

	
\begin{figure}[!t]
      \centering
      \includegraphics[width=2.5 in,height=2.0 in]{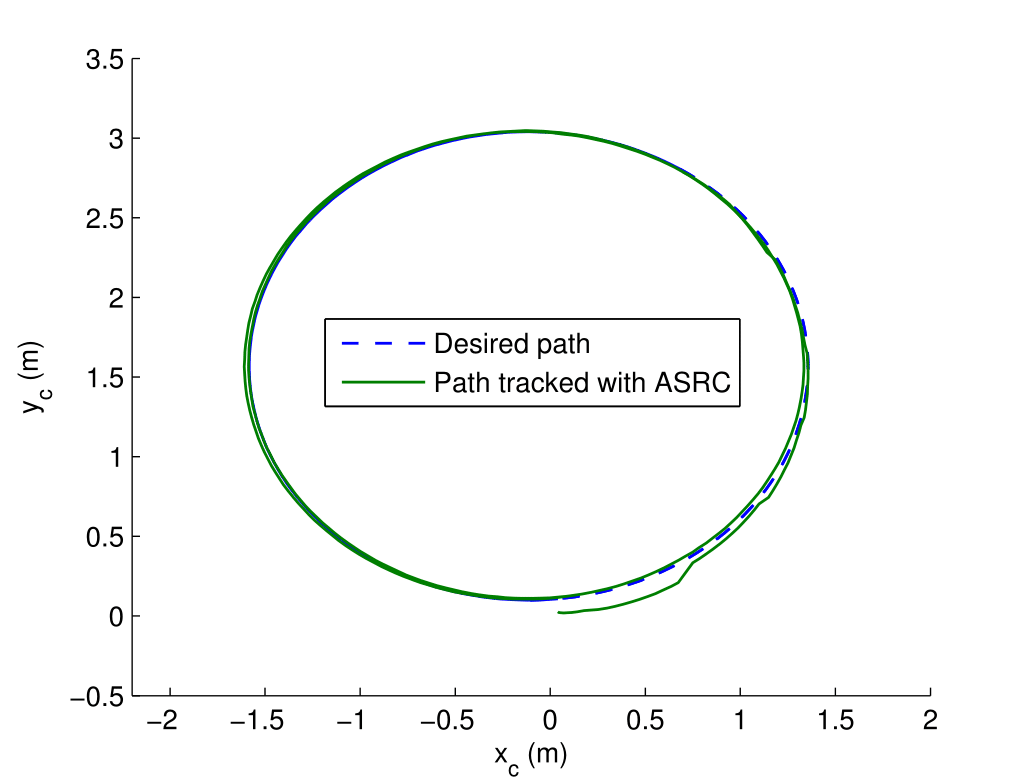}
      \caption{Circular path tracking with ASRC.}
      \label{fig:1}
   \end{figure}

\begin{figure}[!t]
      \centering
      \includegraphics[width=3.5 in,height=2.8 in]{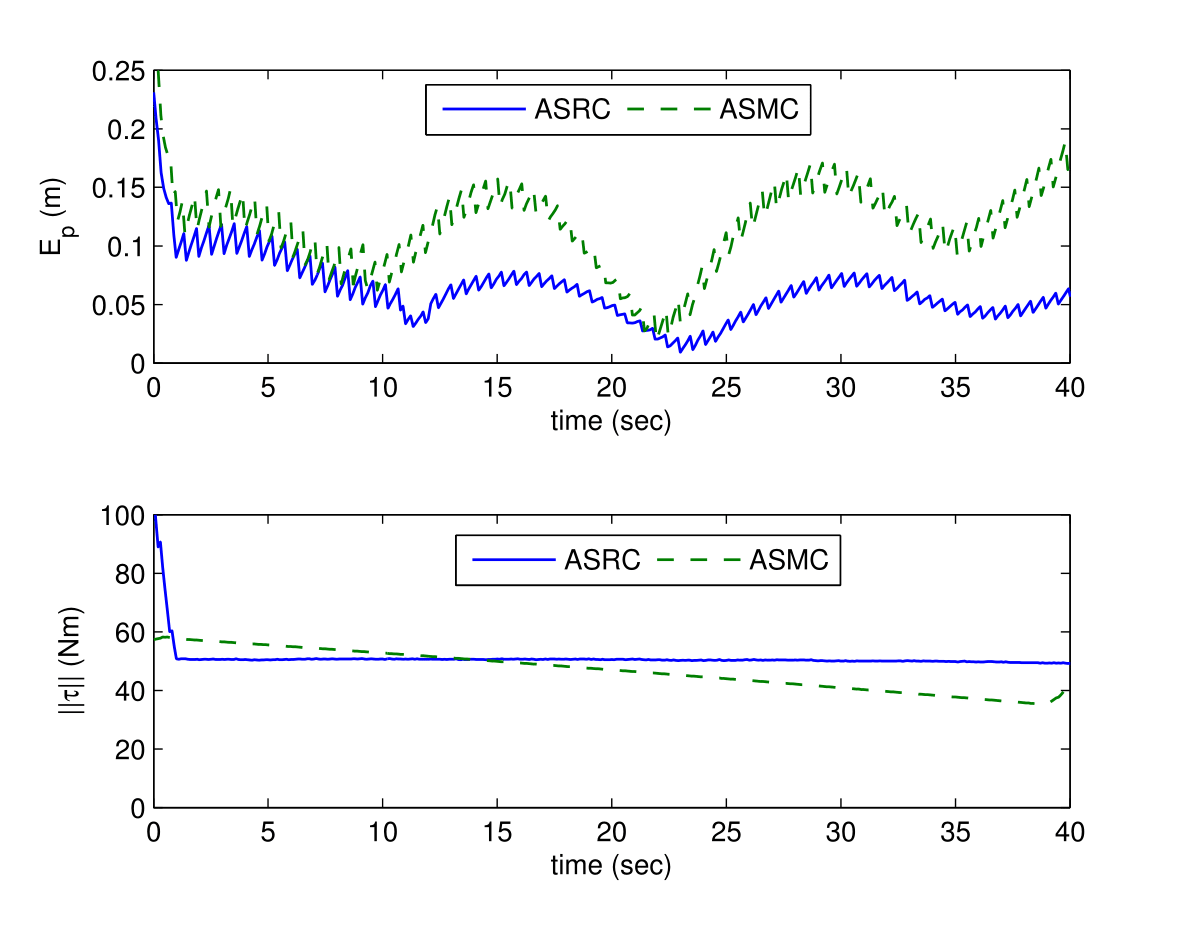} 
      \caption{Performance comparison between ASRC and ASMC.}
      \label{fig:2}
   \end{figure}

\begin{figure}[h!]
      \centering
     \includegraphics[width=3.5 in,height=2.6 in]{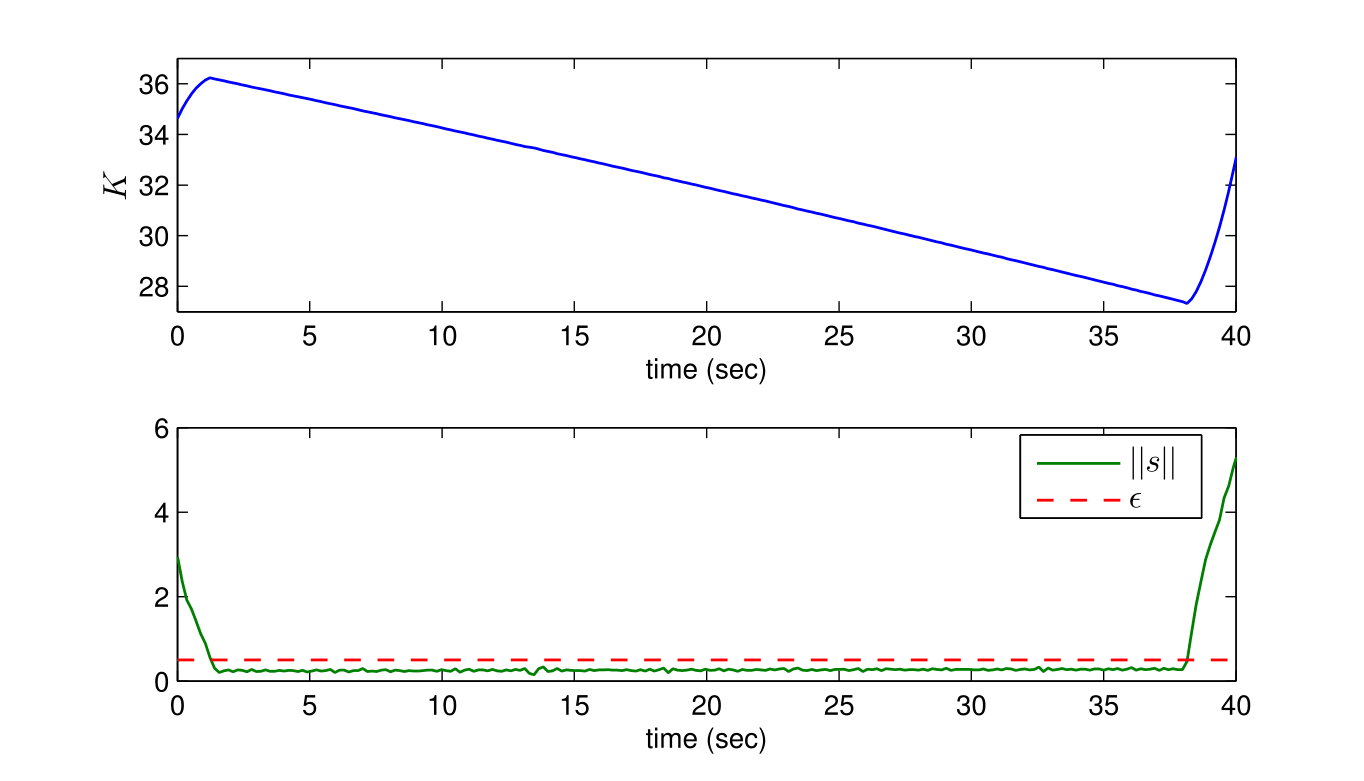} 
      \caption{Switching gain evaluation of ASMC.}
      \label{fig:3}
   \end{figure}

\begin{figure}[h!]
      \centering
       \includegraphics[width=3.5 in,height=3.0 in]{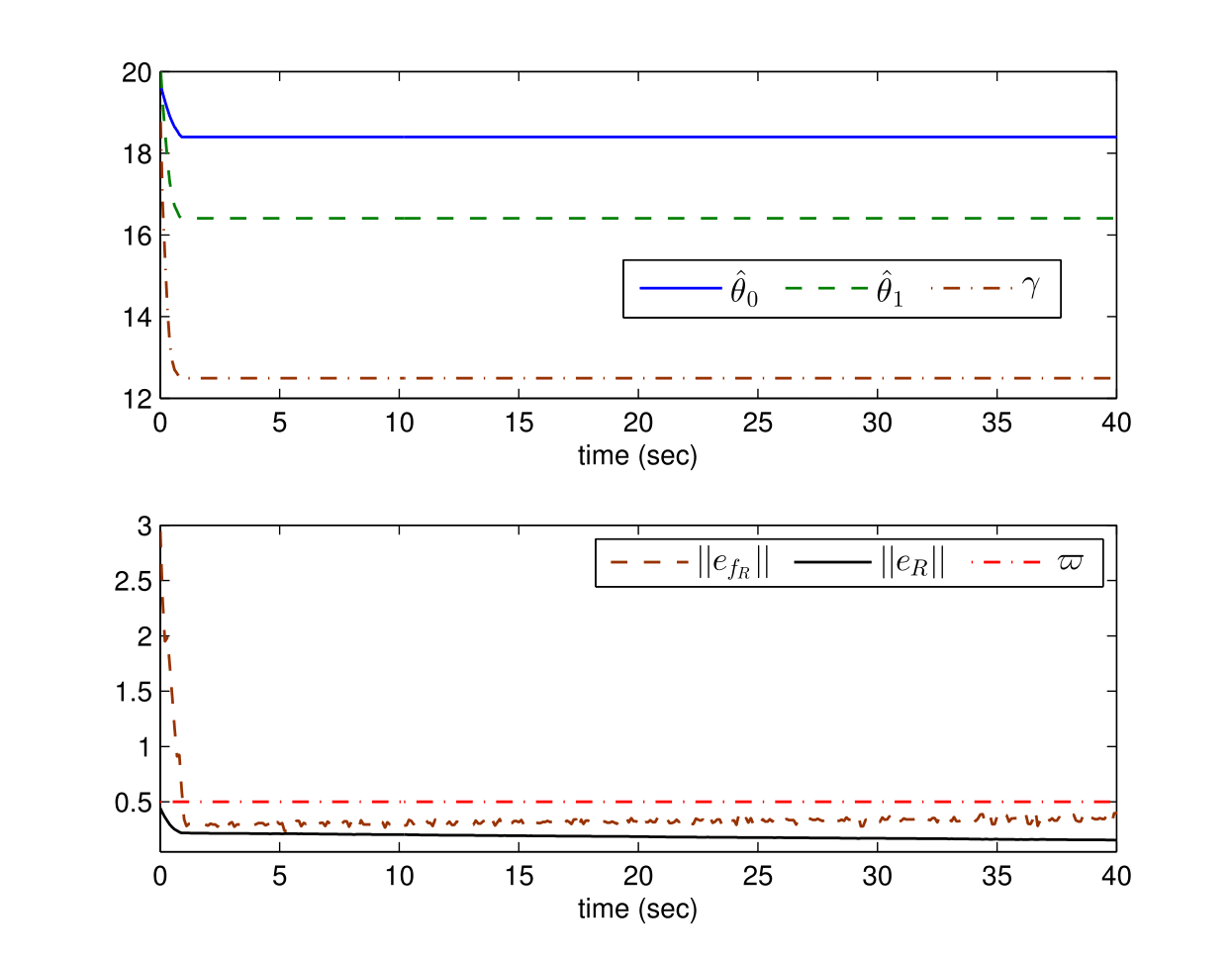}  
      \caption{Switching gain evaluation of ASRC.}
      \label{fig:4}
   \end{figure}
   
   \begin{figure}[!t]
      \centering
       \includegraphics[width=3.5 in,height=2.5 in]{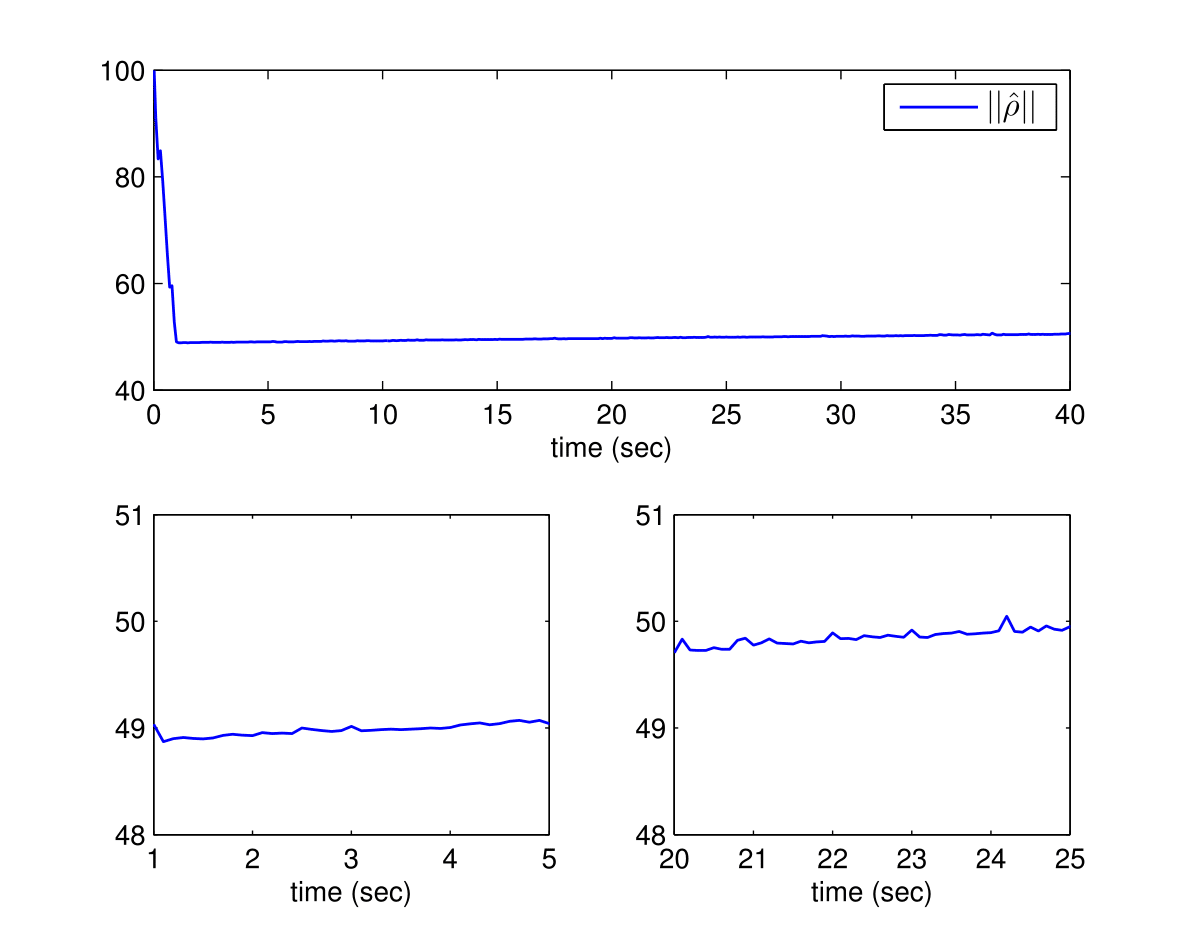}  
      \caption{ { Evaluation of $\hat{\rho}$ of ASRC.}}
      \label{fig:5}
   \end{figure}
   
      \begin{table}[!h]
\renewcommand{\arraystretch}{1.2}
\caption{Performance of ASRC for Various $\varpi$}
\label{table 3}
\begin{center}
\begin{tabular}{c c c c}
\hline
\hline
& $\varpi=0.5$  & $\varpi=0.3$  & $\varpi=0.1$   \\
\hline
RMS (root mean squared) $E_p$ (m) & 0.053 & 0.0421 & 0.0362 \\
RMS $||\tau||$ (Nm) & 89.78 & 75.46 & 67.13 \\
\hline
\hline
\end{tabular}
\end{center}
\end{table}
   
{ It can be noticed from Fig. \ref{fig:5} that initials gains are high enough such that $||e_R||$ decreases from the beginning and so do the gains $\gamma, \hat{\gamma}_i$. Hence, it would be prudent to verify the capability of ASRC in alleviating the overestimation-underestimation problem while starting with relatively low gains. Therefore, the same experiment for ASRC is repeated with much lower initial value of the gains $\hat{\theta}_i$, $i=0,1$ and $\gamma$. Previously, the initial values were $\gamma(t_0)=\hat{\theta}_i(t_0)=20$. This time they are selected to be $\gamma(t_0)=\hat{\theta}_i(t_0)=10$. The tracking performance and evaluation of the switching gain for the later case is shown in Fig. \ref{fig:per_com_new} and \ref{fig:asrc_gain_low_init} respectively. Here ASRC1 denotes when $\gamma(t_0)=\hat{\theta}_i(t_0)=20$ and ASRC2 denotes when $\gamma(t_0)=\hat{\theta}_i(t_0)=10$.

\begin{figure}[!h]
      \centering
      \includegraphics[width=3.5 in,height=2.1 in]{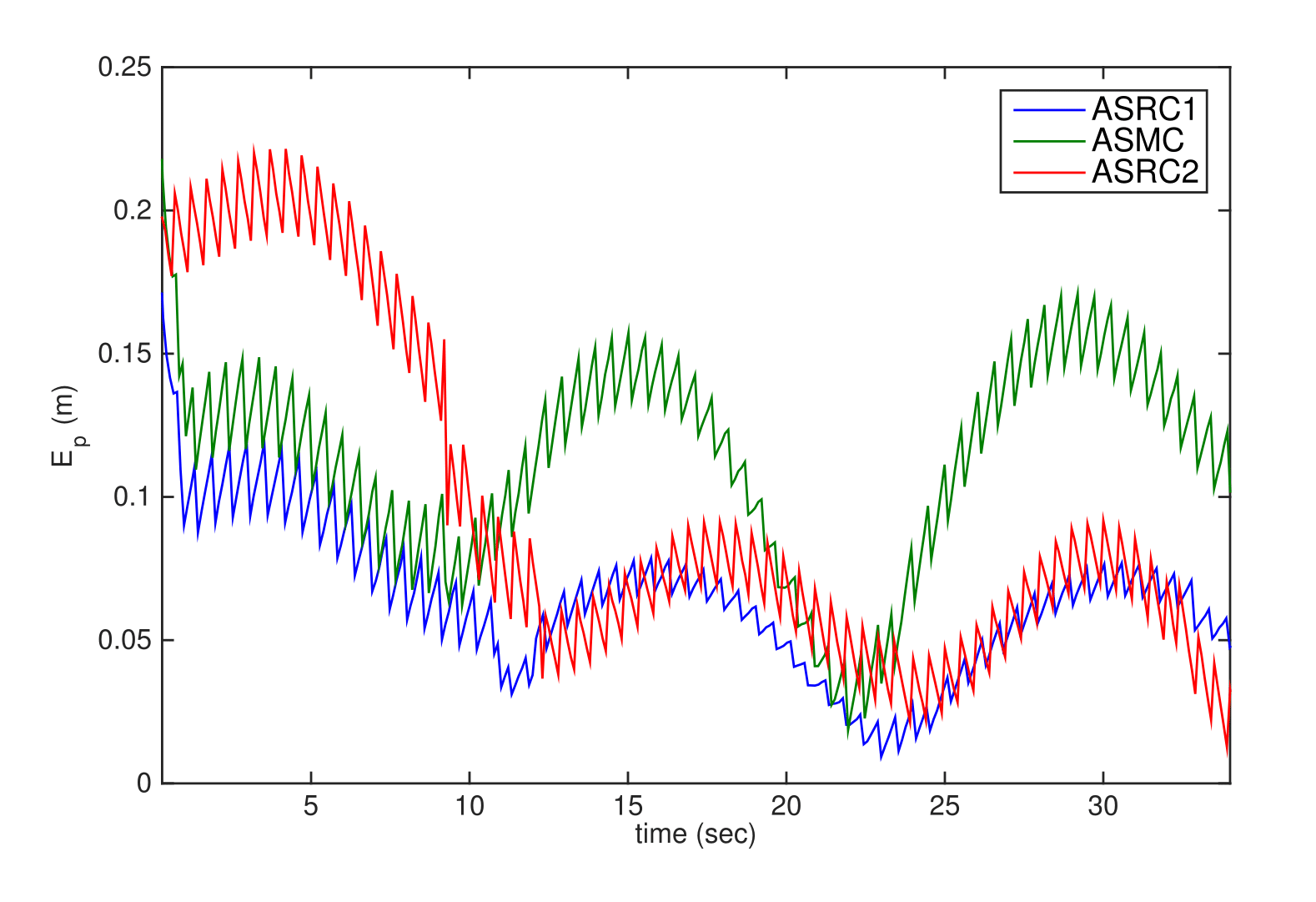} 
      \caption{Performance comparison between various controllers.}
      \label{fig:per_com_new}
   \end{figure}

\begin{figure}[h!]
      \centering
     \includegraphics[width=3.5 in,height=2.5 in]{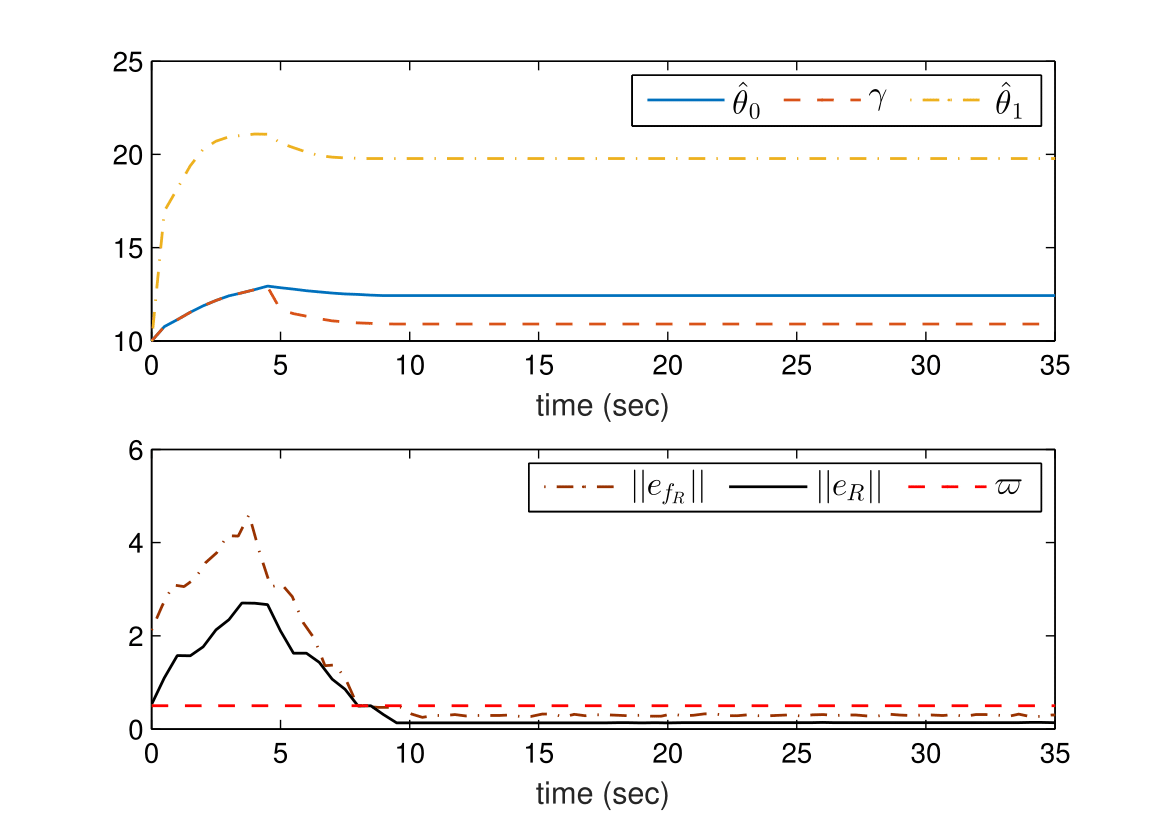} 
      \caption{Switching gain evaluation of ASRC2.}
      \label{fig:asrc_gain_low_init}
   \end{figure}

It can be noticed from Fig. \ref{fig:per_com_new} and \ref{fig:asrc_gain_low_init} that initially the tracking error is high for ASRC2 compared to ASRC1 and ASMC (initial gain $K(t_0)=35$) due to low initial gains. However, at $t \geq 5$sec tracking accuracy of ASRC2 begins to improve as the gains became sufficient enough to negotiate the uncertainties and eventually the tracking performance of ASRC2 is found to be similar to ASRC1 from $t \geq 12$sec and much improved compared to ASMC. This proves that the proposed adaptive law can perform satisfactorily even with low initial conditions of the gains.

Another important aspect to verify is whether ASRC2 can alleviate the over- and under-estimation issue similar to ASRC1. It can verified from Fig. \ref{fig:asrc_gain_low_init} that when $||e_{f_R}|| > \varpi$, the gains follow the pattern of $||e_R||$ according to the adaptive laws (46)-(47). Due to the initial low values, $||e_R||$ increases and so do the gains; similarly at $t \geq 5$sec the gains decrease as $||e_R||$ decreases. Further, at $t \geq 8$sec the gains remain unchanged when they were sufficient to keep the filtered tracking error $||e_{f_R}|| \leq \varpi$, according to the law (48), thus overcoming the underestimation problem. Moreover, that gains do not increase during $t=3.78-8$sec when $||e_R||$ decrease sand thus avoids the overestimation problem. \textbf{Hence, low initial gain conditions do not affect the capability ASRC2 in alleviating the over- and under-estimation problem.}

Moreover, it can be noticed from the first conditions of (46)-(47) that the rate of increment of $\hat{\theta}_0$ and $\gamma$ are same; hence they have similar value in Fig. \ref{fig:asrc_gain_low_init} when $||e_R||$ increases during $t=0-3.7$sec (please note that we have selected $\alpha_0=\alpha_3=10$ in the experiment) . However, their rate of falling are different (second conditions of (46)-(47)); thus, they exhibit different falling pattern in Fig. \ref{fig:asrc_gain_low_init}. }

\section{Conclusions}
A novel ASRC law is proposed for a class of uncertain EL systems where the upper bound of uncertainty possesses a LIP structure. The benefit of the ASRC lies in the fact that it is independent of the nature of the uncertainty and can negotiate uncertainties which can be LIP or NLIP. ASRC does not presume the overall uncertainty to be bounded by a constant and avoids putting any prior restriction on states. Moreover, the proposed adaptive law alleviates the overestimation-underestimation problem of switching gain. The experimental results validate the efficacy of the proposed control law in comparison with the existing adaptive sliding mode control. The future work would be to extend the ASRC law for systems with unmatched disturbances.

\end{document}